%% file: orient.tex
\documentclass[US-letter,USenglish,authorcolumns,autoref]{lipics-v2021}
\input{paper_macros.tex}

\definecolor{winered}{rgb}{0.5,0.2,0}

\usepackage{hyperref}
\hypersetup{ 
    colorlinks = true,
    allcolors={winered},
}

\usepackage{wrapfig}

\usepackage[normalem]{ulem}
\usepackage{changepage}
\usepackage{tabularx}
\usepackage{ragged2e}

\DeclareMathOperator{\firstSymbol}{firstSymbol}
\DeclareMathOperator{\firstNonMin}{firstNonMin}
\DeclareMathOperator{\lastSymbol}{lastSymbol}
\DeclareMathOperator{\lastNonMax}{lastNonMax}
\DeclareMathOperator{\secondLastNonMax}{secondLastNonMax}

\renewcommand{\tt} {\mathtt}

\newcommand{\A} {\mathbf{A}}

\newcommand{\N} {\mathbf{N}}

\newcommand{\Set} {\mathbf{S}}
\newcommand{\OS} {\mathcal{OS}}

\newcommand{\cycletree} {\mathbb{T}}
\newcommand{\treeroot}[1] {\mathit{r}_n}
\newcommand{\ktreeroot}[1] {\mathit{r}_{n,k}}

\DeclareMathOperator{\ap}{ap}

\DeclareMathOperator{\parent}{par}

\nolinenumbers

\title{Constructing $k$-ary Orientable Sequences with Asymptotically Optimal Length}
\titlerunning{~~Orientable sequences}

\author{Daniel Gabri\'{c}}{University of Guelph, Canada}{}{}{}
\author{Joe Sawada}{University of Guelph, Canada}{}{}{}

\author{~}{~}{}{}{}
\authorrunning{~}

\Copyright{~}

\begin{document}
\maketitle

\begin{abstract}
An orientable sequence of order $n$ over an alphabet $\{0,1,\ldots, k{-}1\}$ is a cyclic  sequence such that each length-$n$ substring appears at most once \emph{in either direction}.
When $k= 2$, efficient algorithms are known to construct binary orientable sequences, with asymptotically optimal length, by applying the classic cycle-joining technique. The key to the construction is the definition of a parent rule to construct a cycle-joining tree of asymmetric bracelets.  Unfortunately, the parent rule does not generalize to larger alphabets. Furthermore, 
unlike the binary case, a cycle-joining tree does not immediately lead to a simple successor-rule when $k \geq 3$ unless the tree has certain properties.
In this paper, we derive a parent rule to derive a cycle-joining tree of $k$-ary asymmetric bracelets.  
This leads to 
a successor rule that constructs asymptotically optimal $k$-ary orientable sequences in $O(n)$ time per symbol using $O(n)$ space.  In the special case when $n=2$, we provide a simple construction of $k$-ary orientable sequences of maximal length.
\end{abstract}

\section{Introduction}  \label{sec:intro}

Given a set $\Set$ of $k$-ary strings of length $n$, a \defo{universal cycle} is a cyclic sequence of length $|\Set|$ that contains each string in $\Set$ as a substring exactly once.   When $\Set$ consists of \emph{all} $k$-ary strings of length $n$, universal cycles are known as $\defo{de Bruijn sequences}$. 
Universal cycles have been studied for many fundamental objects including permutations, subsets, and graphs~\cite{godbole2011,chung}.  Universal cycles do not exist directly for permutations; however, efficient constructions exist using a shorthand representation~\cite{shorthand2,shorthand}.  Universal cycles for $n$-subsets of a $k$-set must satisfy the following necessary condition: $k$ divides ${k \choose n}$, or equivalently $n$ divides ${k-1 \choose n-1}$.  For, if $3$-subsets of a $6$-set, $\Set$ will contain exactly one of $\{123, 132, 213, 231, 312, 321\}$. Universal cycles for subsets, are only known to exist for small values of $n$~\cite{hurlbert94,jackson93,rudoy,stevens02}.  This gives rise to studying subset \emph{packings}, that is,  cyclic sequences that contains each $n$-subset \emph{at most} once as a substring~\cite{curtis,debski2016,stevens02}.  

In this paper we are interested in a set $\Set$ that does not contain both a string and its reversal. Similar to the problem of subset packings, determining a maximal set $\Set$ that admits a universal cycle is extremely challenging. 
A universal cycle for such a $k$-ary set of length-$n$ strings is known as an \defo{orientable sequence} of order $n$ (an $\OS_k(n)$).  
By definition, an orientable sequence does not contain a length-$n$ substring that is a palindrome.  

\begin{wrapfigure}[5]{r}{0.4\textwidth}  
\vspace{-0.2in}
        \centering
      \resizebox{2.5in}{!}{\includegraphics{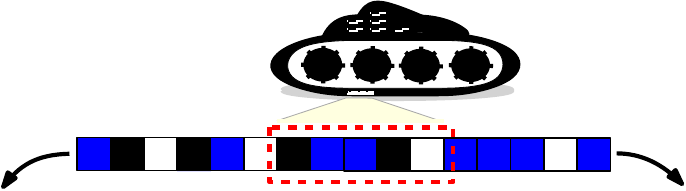}}  
\end{wrapfigure}
\noindent
Orientable sequences were introduced for binary strings by Dai, Martin, Robshaw, and Wild~\cite{Dai} with an application related to robotic position sensing. In particular, consider an autonomous robot with limited sensors.  To determine its location on a cyclic track labeled with coloured squares, the robot scans a window of $n$ squares directly beneath it (see the graphic on the right). For the position \emph{and} orientation to be uniquely determined, the track is designed with the property that each length $n$ window can appear at most once in \emph{either direction}.  


%
%
\begin{exam}  
Consider the sequence $\mathcal{S}= 012013023123$ over the alphabet $\{0,1,2,3\}$.
In the forward direction, including the wraparound, $\mathcal{S}$ contains

\vspace{-0.1in}
\begin{center}
${012}$,
${120}$,
${201}$,
${013}$,
${130}$,
${302}$,
${023}$,
${231}$,
${312}$,
${123}$,
${230}$,
${301}$ 
\end{center}

\vspace{-0.1in}

\noindent
as substrings; 
in the reverse direction $\mathcal{S}$ contains 

\vspace{-0.1in}
\begin{center}
${321}$,
${213}$,
${132}$,
${320}$,
${203}$,
${031}$,
${310}$,
${102}$,
${021}$,
${210}$,
${103}$,
${032}$. 
\end{center}

\vspace{-0.1in}

\noindent
Since each substring is unique, $\mathcal{S}$ is an $\OS_4(3)$ with length 12.  For biological applications, $\{0,1,2,3\}$ represents the four nucleotide bases $\{\tt{A},\tt{C},\tt{G},\tt{T}\}$ of a DNA strand -- see the discussion at the end of the section.
\end{exam} \normalsize

Recently, Mitchell and Wild developed a recursive algorithm  to construct long orientable sequences for a binary alphabet~\cite{MW}. Subsequently,  efficient constructions of $\OS_2(n)$s with asymptotically optimal length have been developed based on cycle-joining~\cite{G&S-Orientable:2024}, and a previously known existence proof~\cite{Dai}.  
However, there is no known construction of long orientable sequences for $k>2$.  In this paper we demonstrate that the binary cycle-joining approach does not naturally generalize to larger alphabets.  In particular, there are special cases for $k=3$ that arise for $n\geq 12$ (see Example~\ref{exam:badone} in Section~\ref{sec:parent}), and there are additional challenges in deriving successor rules from cycle-joining trees for $k\geq 3$ (see Section~\ref{sec:successor}).  However, by deriving a new parent rule that satisfies the Chain Property (see Section~\ref{sec:successor}), we are able to obtain the first efficient construction of long $\OS_k(n)$s.
\begin{result}
\noindent
{\bf Main result}:
 For $k \geq 3$, we develop a successor rule to construct an $\OS_k(n)$ 
  of asymptotically optimal length in $O(n)$ time per symbol using $O(n)$ space.  For $n=2$, we construct $\OS_k(2)$s of maximal length in $O(1)$ time per symbol.
\end{result}

Let $M_k(n)$ denote the maximum length of an $\OS_k(n)$.
When $k=2$, the maximum length of an orientable sequence is known only for $n\leq 7$~\cite{Dai,G&S-Orientable:2024}.  For $n=2$ and $k\geq 3$, we demonstrate that $M_k(2) = k\lfloor (k-1)/2\rfloor$ by a simple construction (see Section~\ref{sec:max2}).  For $n\geq 3$, exhaustive search demonstrates that $M_3(3) = 9$, $M_4(3) = 20$, and $M_3(4) = 30$.  Search also reveals an $\OS_5(3)$ of length 50 which attains the upper bound stated in~\cite{Alhakim&etal:2023}, and thus $M_5(3) = 50$. Orientable sequences that admit these maximal lengths are given below: 
\begin{itemize}
    \item $n=3$, $k=3$: 001120122~~ (9), 
    \item $n=3$, $k=4$: 00112012230130231233~~ (20),
    \item $n=3$, $k=5$: 00112003102210320331140142042132143043144223342344~~  (50), 
    \item $n=4$, $k=3$: 000102001201112022101121022212~~ (30).
\end{itemize}

\noindent
Since the number of palindromes of length $n$ is $k^{\lfloor (n+1)/2\rfloor}$, a trivial upper bound on $M_k(n)$  is $(k^n - k^{\lfloor (n+1)/2\rfloor})/2$. A deeper analysis on upper bounds is given in~\cite{Alhakim&etal:2023}.


Recall the problem of determining a robot's position and orientation on a track. Suppose now that we allow the track to be non-cyclic. 
The corresponding sequence that allows one to determine orientation and position is called an \emph{acyclic orientable sequence}. 
One can construct an acyclic $\OS_k(n)$ from a cyclic $\OS_k(n)$ by taking the cyclic $\OS_k(n)$ and appending its prefix of length $n{-}1$ to the end. See the paper by Burns and Mitchell~\cite{BM} for more on binary acyclic orientable sequences, which they call \emph{aperiodic $2$-orientable window sequences}.  Gabric and Sawada provide some long acyclic orientable sequences in~\cite{OrientJournal}.
%
%
Rampersad and Shallit~\cite{Rampersad&Shallit:2003} showed that for every alphabet size $k\geq 2$ there is an infinite sequence such that for every sufficiently long substring, the reversal of the substring does not appear in the sequence. Fleischer and Shallit~\cite{Fleischer&Shallit:2019} later reproved the results of the previous paper using theorem-proving software. See~\cite{Currie&Lafrance:2016, Mercas:2017} for more work on sequences avoiding reversals of substrings.

Families of strings related to orientable sequences also appear in DNA computing. Two single strands of DNA can bind to each other if they are ``reverse complements'' of each other, where $\tt{A}$ is the complement of $\tt{T}$ and $\tt{C}$ of $\tt{G}$. The binding of DNA strands allows for the creation of secondary structures, which are useful in certain DNA computing techniques~\cite{PRS:1998}. For example, a \emph{stem-loop}, also known as a \emph{hairpin}, is a DNA secondary structure that has applications in DNA computing~\cite{Domaratzki:2006,Kari&etal:2005,Takinoue&Suyama:2006}. Roughly speaking, a string of symbols $u$ contains a hairpin if it has substring $v$ and $\theta(v)$, where $\theta$ is an antimorphic involution. Hairpin-free strings, that is, strings that do not contain a hairpin of sufficient length, have been studied with the motivation of creating a large collection of DNA molecules that do not bind to themselves in undesirable ways~\cite{KKS:2005}. In the case that $\theta$ is the mirror involution (i.e., reversal), hairpin-free strings are essentially orientable sequences without the restriction that every substring of a specified length occurs at most once. Thus, every orientable sequence is a hairpin-free sequence. See~\cite{KKLST:2006,TKYO:2005} for more on applications of long hairpin-free sequences.

\medskip

\noindent {\bf Outline.}  
In Section~\ref{sec:prelim}, we provide background definitions and notation, including a review of the cycle-joining technique, and $k$-ary successor rules.
In Section~\ref{sec:max2}, we present a simple construction for $\OS_k(2)$s and demonstrate they are of maximal length.
In Section~\ref{sec:symmetric}, we define $k$-ary symmetric/asymmetric necklaces and bracelets, and provide some useful properties of these objects.
In Section~\ref{sec:parent}, we provide a parent rule for constructing a cycle-joining tree composed of asymmetric bracelets. This leads to an $O(n)$ time per symbol successor-rule construction of $\OS_k(n)$s that we demonstrate has asymptotically optimal length in Section~\ref{sec:bounds}.
An implementation of our construction is available for download at \url{http://debruijnsequence.org/db/orientable}.

\section{Preliminaries} \label{sec:prelim}

Let $\Sigma = \{0,1,2,\ldots, k{-}1\}$ be an alphabet of size $k \geq 2$. 
Let $\Sigma^n$ denote the set of all length-$n$ strings over $\Sigma$. 
Let 
 $\alpha = \tt{a}_1\tt{a}_2\cdots \tt{a}_n \in \Sigma^n$ and $\beta = \tt{b}_1\tt{b}_2\cdots \tt{b}_m \in \Sigma^m$ for some $m,n\geq 0$, 
 Throughout this paper, we use lexicographic order when comparing two  strings. More specifically, 
  $\alpha < \beta$  if $\alpha$  is a prefix of $\beta$ or if $\tt{a}_i < \tt{b}_i$ for the smallest $i$ such that $\tt{a}_i \neq \tt{b}_i$. 
Let $\alpha^R$ denote the reversal $\tt{a}_n\tt{a}_{n-1}\cdots \tt{a}_1$ of $\alpha$; $\alpha$ is a \defo{palindrome} if $\alpha = \alpha^R$.
For $j\geq 1$, let $\alpha^j$ denote $j$ copies of $\alpha$ concatenated together.
If $\alpha = \gamma^j$ for some non-empty string $\gamma$ and some $j > 1$, then $\alpha$ is said to be \defo{periodic}; 
otherwise, $\alpha$ is said to be \defo{aperiodic} (or \defo{primitive}). Let $\ap(\alpha)$ denote the shortest string $\gamma$ such that $\alpha = \gamma^t$ for some positive integer $t$; we say $\gamma$ is the \defo{aperiodic prefix} of $\alpha$.  

A \defo{necklace class} is an equivalence class of strings under rotation. Let $[\alpha]$ denote the set of strings in $\alpha$'s necklace class.  We say $\alpha$ is a \defo{necklace} if it is the lexicographically smallest string in $[\alpha]$. Let $\tilde{\alpha}$ denote the necklace in $[\alpha]$.  For example, if $\alpha = 0201$, then $[\alpha] = \{0201, 2010, 0102, 1020\}$ and $\tilde \alpha = 0102$.
Let $\N_k(n)$ denote the set of $k$-ary necklaces of length $n$.
A \defo{bracelet class} is an equivalence class of strings under rotation and reversal.
We say $\alpha$ is a \defo{bracelet} if it is the lexicographically smallest string in $[\alpha] \cup [\alpha^R]$. A bracelet is always a necklace, but a necklace need not be a bracelet. 

Given $\mathbf{S} \subseteq \Sigma^n$, a \defo{universal cycle} $U$ for $\mathbf{S}$ is a cyclic sequence of length $|\mathbf{S}|$ that contains each string in $\mathbf{S}$ as a substring (exactly once).  An \defo{orientable sequence} is a universal cycle where if $\alpha \in \Set$, then $\alpha^R \notin \Set$.
If $\Set = \Sigma^n$ then $U$ is known as a \defo{de Bruijn sequence}.    
Given a universal cycle $U$ for a set $\mathbf{S} \subseteq \Sigma^n$, a \defo{successor rule} for $U$ is a function $f:\mathbf{S} \rightarrow \Sigma$ such that $f(\alpha)$ is the symbol following $\alpha$ in $U$.  

\subsection{Cycle-joining trees}  \label{sec:cycle-join}

In this section we review how two universal cycles can be joined  to obtain a larger universal cycle. Cycle joining is perhaps the most fundamental technique applied to construct universal cycles;  it
has graph-theoretic underpinnings related to Hierholzer's algorithm for constructing Euler cycles~\cite{hierholzer}.
For some applications, see~\cite{etzion-cutting,Etzion1987,EtzionPSR,fred-nfsr,karyframework,huang,jansen,multi,weakorder}.

Let $\tt{x},\tt{y}$ be distinct symbols in $\Sigma$.
If $\alpha = \tt{x}\tt{a}_2\cdots \tt{a}_n$ and $\hat \alpha = \tt{y}\tt{a}_2\cdots \tt{a}_n$, then $\alpha$ and $\hat \alpha$ are said to be \defo{conjugates} of each other, and $(\alpha, \hat \alpha)$ is called a \defo{conjugate pair}.
We say $\gamma$ \defo{belongs to} a conjugate pair $(\alpha, \hat \alpha)$ if either $\gamma = \alpha$ or $\gamma = \hat \alpha$. 
The following well-known result (see for instance Lemma 3 in~\cite{dbrange}) based on conjugate pairs is the crux of the cycle-joining approach.

\begin{theorem} \label{thm:concat}
Let $\mathbf{S}_1$ and $\mathbf{S}_2$ be disjoint subsets of $\Sigma^n$
such that $\alpha = \tt{x}\tt{a}_2\cdots \tt{a}_n \in \mathbf{S}_1$ and $\hat \alpha = \tt{y}\tt{a}_2\cdots \tt{a}_n \in \mathbf{S}_2$; $(\alpha, \hat \alpha)$ is a conjugate pair. 
If $U_1$  is a universal cycle for $\mathbf{S}_1$ with suffix $\alpha$ and $U_2$ is a universal cycle for $\mathbf{S}_2$  with suffix $\hat \alpha$ then $U = U_1U_2$ is a universal cycle for $\mathbf{S}_1 \cup \mathbf{S}_2$.
\end{theorem}

\noindent
Let $U_i$ denote a universal cycle for $\mathbf{S}_i \subseteq \Sigma^n$. Two universal cycles $U_1$ and $U_2$ are said to be \defo{disjoint} if $\mathbf{S}_1 \cap \mathbf{S}_2 = \emptyset$.
Theorem~\ref{thm:concat} states that
two disjoint universal cycles can be joined to form a single universal cycle if they each contain one string of a conjugate pair as a substring. Note that necklaces correspond to disjoint cycles that partition the set $\Sigma^n$.
A \defo{cycle-joining tree} $\cycletree$  is an unordered tree where 
the nodes correspond to a disjoint set of universal cycles $U_1, U_2, \ldots, U_t$; an edge between $U_i$ and $U_j$ is defined by a conjugate pair $(\alpha, \hat \alpha)$ such that $\alpha \in \mathbf{S}_i$ and $\hat \alpha \in \mathbf{S}_j$.  For our purposes, we consider cycle-joining trees to be rooted. 
%
\begin{exam}  \label{exam:UC}
Let $n=3$ and $k=4$.
Consider necklace classes $\Set_1 = [021]$, $\Set_2 = [011]$, and $\Set_3 = [031]$
\begin{wrapfigure}{r}{0.2\textwidth}
\centering  \vspace{-0.1in}
\resizebox{0.7in}{!}{\includegraphics{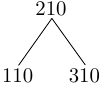}}
\end{wrapfigure}
with corresponding universal cycles $U_1 = 210$, $U_2 = 110$ and $U_3 = 310$.  
The three cycles can be joined via conjugate pairs ($210,110$) and ($210,310$)
to form the cycle-joining tree on the right.
Joining $U_1$ and $U_2$ we obtain the larger cycle $210110$; joining $U_3$ to this cycle we obtain the universal cycle $110210310$ for 
$(\Set_1 \cup \Set_2) \cup \Set_3$.  If we join $U_1$ and $U_3$ first, we obtain a different universal cycle $310210110$.
%
\end{exam} \normalsize
\noindent
Many universal cycle constructions have a corresponding cycle-joining tree that can be defined by a rather simple \defo{parent rule}. For example, when $\Set = \Sigma^n$ and $\alpha = \tt{a}_1\tt{a}_2\cdots \tt{a}_n \in \N_k(n)$, the following are four of the \emph{simplest} parent rules that define how to construct cycle-joining trees with nodes corresponding to $\N_k(n)$~\cite{karyframework}:
\begin{itemize}
    \item $\firstSymbol(\alpha) = $ the necklace of $[\blue{(\tt{a}_1{-}1)} \tt{a}_2\cdots \tt{a}_n]$ with root $(k{-}1)^n$,
    \item $\lastSymbol(\alpha) = $ the necklace of $[\tt{a}_1\tt{a}_2 \cdots \tt{a}_{n-1}\blue{(\tt{a}_n{+}1)}]$ with root $0^n$,  
    \item $\firstNonMin(\alpha) = 0^{i-1}\blue{(\tt{a}_i{-}1)}\tt{a}_{i+1}\cdots \tt{a}_n$ with root $0^n$, and
    \item $\lastNonMax(\alpha) = \tt{a}_1\cdots\tt{a}_{j-1}\blue{(\tt{a}_j{+}1)}(k{-}1)^{n-j}$ with root $(k{-}1)^n$,
\end{itemize}
where $i$ denotes the index of the first non-zero in $\alpha$ and  $j$ denotes the index of the last non-($k{-}1$) in $\alpha$.  Addition on the symbols is modulo $k$.  From the definition of a necklace, it is straightforward to see that if $\alpha$ is a (non-root) necklace, then  both $0^{i-1}\blue{(\tt{a}_i{-}1)}\tt{a}_{i+1}\cdots \tt{a}_n$ and $\tt{a}_1\cdots\tt{a}_{j-1}\blue{(\tt{a}_j{+}1)}(k{-}1)^{n-j}$ are also necklaces.  
Note that if a necklace $\alpha$ is periodic, the corresponding universal cycle for $\alpha$ is its aperiodic prefix.  However, for simplicity of understanding, we use the full necklace to represent a node in our cycle-joining trees.  For instance, when $n=3$, we use $000$ instead of $0$.
As an example, Figure~\ref{fig:pcr33} illustrates the cycle-joining trees induced by the four parent rules above with nodes $\N_3(3)$.  Each node $\alpha$ and its parent $\beta$ are joined by a conjugate pair, where the highlighted bit in $\alpha$ is the first bit in one of the conjugates.  When $k > 3$, we apply the last three of the four parent rules to construct an $\OS_k(n)$ with length that is asymptotically optimal.  When $k=3$, we must introduce one additional function.

\begin{figure}[ht]
\centering
\resizebox{5.7in}{!}{\includegraphics{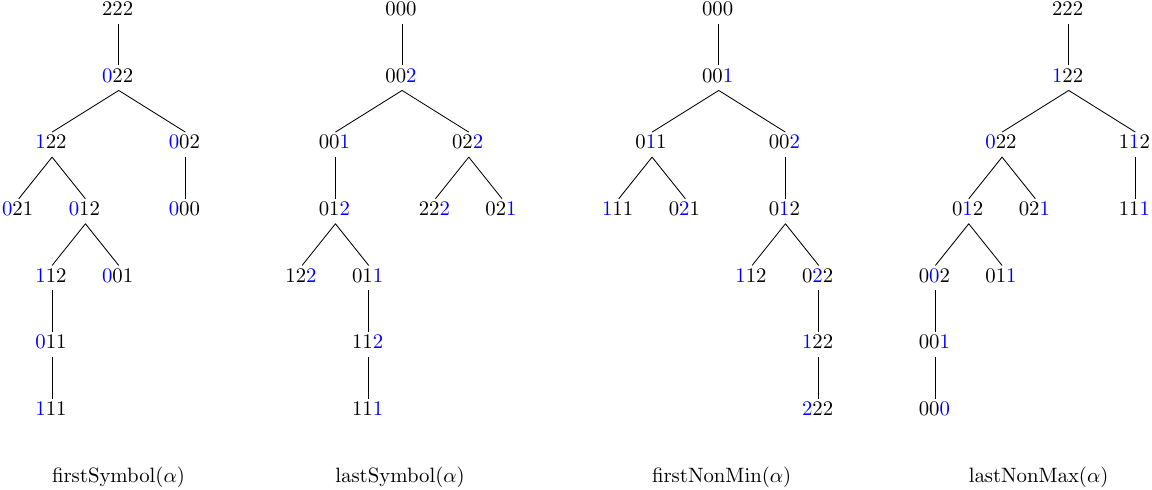}}
\caption{
Cycle-joining trees for $\N_3(3)$ induced by four different parent rules.
}
\label{fig:pcr33}
\end{figure}


\subsection{Successor rules} \label{sec:successor}

In this section we outline how to derive a successor-rule from a cycle-joining tree $\cycletree$ for an underlying set $\Set$.  In the binary case, each cycle-joining tree corresponds to a unique universal cycle; however, when $k>2$, this is not necessarily the case.  

\begin{result}
\noindent
{\bf Uniqueness Property}:  
The cycles in a cycle-joining tree $\cycletree$ are joined such that no two conjugate pairs that have a string in common, i.e., there are no two conjugate pairs of the form ($\tt{x}\tt{a}_2\cdots \tt{a}_n, \blue{\tt{y}}\tt{a}_2\cdots \tt{a}_n$) and  
($\tt{x}\tt{a}_2\cdots \tt{a}_n, \red{\tt{z}}\tt{a}_2\cdots \tt{a}_n$).
\end{result}
The Uniqueness Property is always satisfied when $k=2$, but is not necessarily the case for $k>2$. For example, none of the trees in Figure~\ref{fig:pcr33} satisfy the Uniqueness Property.   If $\cycletree$ has the Uniqueness Property, then a successor rule for the corresponding \emph{unique} universal cycle is given by $f(\alpha)$, where $\alpha = \tt{a}_1\tt{a}_2\cdots \tt{a}_n$:
%
\begin{center}
$f(\alpha) = \left\{ \begin{array}{ll}
         \tt{y} &\ \  \mbox{if $\alpha$ belongs to some conjugate pair $(\alpha, \blue{\tt{y}}\tt{a}_2\cdots \tt{a}_n)$;}\\
         \tt{a}_1 \  &\ \  \mbox{otherwise.}\end{array} \right. $
\end{center}
%

\noindent
When the Uniqueness Property is not satisfied, a universal cycle derived from a cycle-joining tree depends on the order that the cycles are joined together, as illustrated in Example~\ref{exam:UC}.  See~\cite{concat} for a deeper analysis on different universal cycles that can be obtained from the same cycle-joining tree. 
 Our challenge is to create a relatively simple successor rule that defines a universal cycle derived from a cycle-joining tree $\cycletree$ that does not satisfy the Uniqueness Property. 
Ultimately, we require our cycle-joining tree to be defined with the following property:
%
\begin{result}
\noindent
{\bf Chain Property}: If a node in a cycle-joining tree $\cycletree$ has two children joined via conjugate pairs $(\tt{x}\tt{a}_2\cdots \tt{a}_n, \tt{y}\tt{a}_2\cdots \tt{a}_n)$ and 
$(\tt{x}'\tt{b}_2\cdots \tt{b}_n, \tt{y'}\tt{b}_2\cdots \tt{b}_n)$, then $\tt{a}_2\cdots \tt{a}_n \neq \tt{b}_2\cdots \tt{b}_n$. 
\end{result}
Let $\alpha_1, \alpha_2, \ldots, \alpha_m$ denote a maximal-length path of nodes in $\cycletree$ such that for each $1 \leq i < m$, the node $\alpha_i$ is the parent of $\alpha_{i+1}$
and they are joined via a conjugate pair of the form  $(\tt{x}_i\beta, \tt{x}_{i+1}\beta$), where $\beta$ is the same in each conjugate pair.  
We call such a path a \defo{chain} of length $m$.  
Any node with $j$ children will belong to at least $j$ chains.
Given a chain let $\Call{Next}{\tt{x}_i\beta} = \tt{x}_{i+1}$, where $\tt{x}_{m+1} = \tt{x}_1$.
%
\begin{exam} 
The nodes $\alpha_1 = 112$, $\alpha_2 = 011$, and $\alpha_3 = 111$ in the first tree of Figure~\ref{fig:pcr33} form a chain with length $m=3$ joined by conjugate pairs (211,011) and (011,111).  $\Call{Next}{211} = 0$, $\Call{Next}{011} = 1$, and $\Call{Next}{111} = 2$.
\end{exam} \normalsize
\noindent

%
If $\cycletree$ is a cycle-joining tree  with the Chain Property for an underlying set $\Set$, then the following function $g$ is a successor rule for a universal cycle of $\Set$ (based on theory in~\cite{karyframework}):
%
\begin{center}
$g(\alpha) = \left\{ \begin{array}{ll}

 \Call{Next}{\alpha}  &\ \  \mbox{if $\alpha$ belongs to some conjugate pair;}\\
         {\tt{a}_1} \  &\ \  \mbox{otherwise.}\end{array} \right. $
         
\end{center}
When $k=2$, $f=g$.  As stated, this successor rule requires exponential space to store the conjugate pairs.  In our application (see Section~\ref{sec:succ}), the corresponding successor rule will run in $O(n)$ time per symbol and use $O(n)$ space.

%

\section{A maximal-length construction for $\OS_k(2)$s} \label{sec:max2}

In this section we consider the case when $n=2$.  There does not exist an $\OS_1(2)$ or an $\OS_2(2)$, so we assume $k \geq 3$. Since there are no substrings of the form $\tt{x}\tt{x}$ in any $\OS_2(2)$, each symbol in any $\OS_k(2)$ can appear at most $\lfloor (k-1)/2 \rfloor$ times. Thus, $M_k(2) \leq k \lfloor (k-1)/2 \rfloor$. In fact, we show that this bound is tight via a simple construction that depends on the parity of $k$.

For each pair of distinct symbols $\tt{x}$ and $\tt{y}$, an $\OS_k(2)$ can contain either $\tt{xy}$ or $\tt{y}\tt{x}$, but not both. If $k$ is odd, our $\OS_k(2)$ will contain exactly one such substring for each pair of symbols.  If $k$ is even, we must remove $k/2$ such pairs to meet the upper bound. In particular, we remove the pairs $\{0,1\}, \{2,3\}, \ldots , \{k{-}2,k{-}1\}$.  We outline our choices from each pair of symbols,  and illustrate $\OS_k(2)$s in Figure~\ref{fig:paths} for $k=7$ and $k=8$; it is straightforward to generalize the construction depending on the parity of $k$, as follows.
\begin{itemize}
\item {\bf $k$ odd}.  Let $\sigma_3 = 012$. For $k \geq 5$,  let $\sigma_k = \tt{s}_1\tt{s}_2\cdots \tt{s}_{2k-3}$  where
$\tt{s}_1\tt{s}_3\tt{s}_5 \cdots \tt{s}_{2k-5} = 012\cdots (k{-}3)$ and
$\tt{s}_2\tt{s}_4\tt{s}_6 \cdots \tt{s}_{2k-4}\tt{s}_{2k-3} = ((k{-}2)(k{-}1))^{(k-1)/2}$.  \medskip

\item {\bf $k$ even}.  Let $\sigma_4 = 0213$. For $k \geq 6$, let  $\tau_k = \tt{t}_1\tt{t}_2\cdots \tt{t}_{2k-4}$  where
$\tt{t}_1\tt{t}_3\tt{t}_5 \cdots \tt{t}_{2k-5} = 012\cdots (k{-}3)$ and
$\tt{t}_2\tt{t}_4\tt{t}_6 \cdots \tt{t}_{2k-4} = ((k{-}2)(k{-}1))^{k/2-1}$.  
\end{itemize}

\noindent
Let $U_k = \sigma_3\sigma_5 \cdots \sigma_k$ for $k$ odd,  and let 
 $U_k = \tau_4\tau_6 \cdots \tau_k$ for $k$ even.  
 
\begin{figure}[ht]
    \centering
    \resizebox{4.5in}{!}{\includegraphics{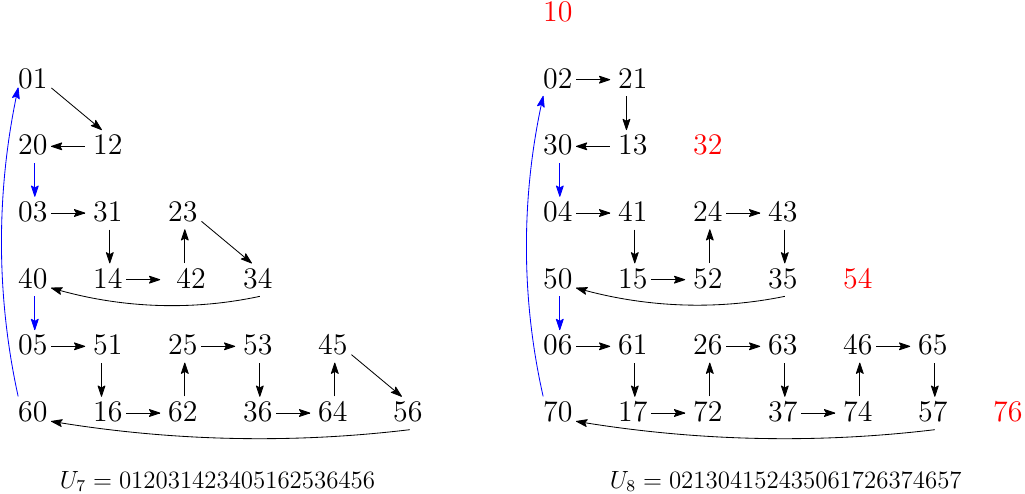}}
    \caption{Illustrating the construction of $U_7$ and $U_8$.}
    \label{fig:paths}
\end{figure}

\begin{exam}
    The following illustrates $U_k$ for $3 \leq k \leq 8$:
    \medskip
    
    \begin{minipage}{3in}
        \begin{itemize}
    \item $U_3 =$  012~~
    \item $U_5 =$  012~0314234~~
    \item $U_7 =$  012~0314234~05162536456~~ 

\end{itemize}
    \end{minipage}~~~~
    \begin{minipage}{3in}
        \begin{itemize}
    \item $U_4 =$  0213~~ 
    \item $U_6 =$  0213~04152435~~ 
    \item $U_8 = $ 0213~04152435~061726374657 
\end{itemize}
    \end{minipage}
\end{exam}

\begin{theorem}
For $k > 2$,  $U_k$ is an $\OS_k(2)$ with length $k \lfloor (k-1)/2 \rfloor$ that can be generated in $O(1)$-time per symbol.
\end{theorem}
\begin{proof}
We provide a high-level proof outline.
Clearly $U_3$ is an $\OS_3(2)$ with length $3$, and $U_4$ is an $\OS_4(2)$ with length 4.
Suppose $k \geq 5$ is odd.  The sequences
$\sigma_3, \sigma_5, \cdots, \sigma_k$ are disjoint $\OS_k(2)$s that are cycle-joined to obtain $U_k$. Simple math shows $U_k$ has length $k\lfloor (k-1)/2 \rfloor$. A similar analysis can be applied for even $k$. It is straightforward to generate $U_k$ in constant time per symbol.
\end{proof}

\noindent
An immediate consequence of this theorem is
$M_k(2) = k \lfloor (k-1)/2 \rfloor$.

Observe that when $n$ is odd, $U_k$ is a universal cycle for the $2$-subsets of a $k$-set. Previously, such universal cycles were claimed to exist in~\cite{chung}, with justification in~\cite{jackson93}.  No construction was previously provided.  When $n$ is even, $U_k$ is a maximal $2$-subset packing of a $k$-set.

\section{Symmetric and asymmetric necklaces and bracelets} \label{sec:symmetric}

In this section, we present properties for symmetric/asymmetric necklaces and bracelets that are necessary for our main results in Section~\ref{sec:parent}.
A necklace $\alpha$ is \defo{symmetric} if it belongs to the same necklace class as $\alpha^R$, i.e., both $\alpha$ and $\alpha^R$ belong to $[\alpha]$.  By this definition, a symmetric necklace is a bracelet. If a necklace or bracelet is not symmetric, it is said to be \defo{asymmetric}.
Let $\A_k(n)$ denote the set of all $k$-ary asymmetric bracelets of length $n$.  
Table~\ref{tab:neck} lists all $70$ necklaces in
$\N_4(4)$ partitioned into asymmetric necklace pairs and symmetric necklaces. The asymmetric necklace pairs belong to the same bracelet class, and the first string in each pair is an asymmetric bracelet.  Thus, $|\A_4(4)| = 15$.
In general, $|\A_k(n)|$ is equal to the number of $k$-ary necklaces of length $n$ minus the number of $k$-ary bracelets of length $n$.  For more background
on asymmetric bracelets, see~\cite{brace,G&S-Orientable:2024}.

\begin{table}[ht]
    \centering

\begin{tabular}{l | l }
{\bf Asymmetric necklace pairs} & ~~~~~~~~~ {\bf Symmetric necklaces} \\ \hline

~~\blue{0012}~,~0021  ~~~~~  \blue{0133}~,~0331  &            0000            ~~~ 010\underline{2} ~~~ 0\underline{222} ~~~  111\underline{3} ~~~ 1\underline{323} \\ 
~~\blue{0013}~,~0031  ~~~~~   \blue{0213}~,~0312  &         000\underline{1}  ~~~ 010\underline{3} ~~~  0\underline{232} ~~~  11\underline{22} ~~~ 1\underline{333}    \\ 
~~\blue{0023}~,~0032   ~~~~~ \blue{0223}~,~0322  &          000\underline{2}  ~~~ 0\underline{111}  ~~~ 030\underline{3} ~~~  11\underline{33} ~~~ 2222 \\  
~~\blue{0112}~,~0211  ~~~~~   \blue{0233}~,~0332  &         000\underline{3}  ~~~  0\underline{121} ~~~ 0\underline{313} ~~~  121\underline{2} ~~~ 222\underline{3} \\  
~~\blue{0113}~,~0311  ~~~~~  \blue{1123}~,~1132  &           00\underline{11} ~~~ 0\underline{131}  ~~~ 0\underline{323} ~~~  121\underline{3} ~~~ 22\underline{33} \\ 
~~\blue{0122}~,~0221  ~~~~~  \blue{1223}~,~1322  &          00\underline{22}  ~~~ 020\underline{2} ~~~  0\underline{333} ~~~ 1\underline{222} ~~~ 232\underline{3} \\ 
~~\blue{0123}~,~0321  ~~~~~  \blue{1233}~,~1332  &           00\underline{33} ~~~ 020\underline{3} ~~~ 1111               ~~~ 1\underline{232}   ~~~ 2\underline{333}\\  
~~\blue{0132}~,~0231                        &         010\underline{1}    ~~~ 0\underline{212} ~~~ 111\underline{2}  ~~~ 131\underline{3}  ~~~   3333\\ 

\end{tabular}

    \caption{The $70$ necklaces in $\N_4(4)$ partitioned into 15 asymmetric necklace pairs and 40 symmetric necklaces. The first (highlighted) necklace in each pair is an asymmetric bracelet in $\A_4(4)$.}
    \label{tab:neck}
\end{table}

Let $\alpha_1, \alpha_2, \ldots, \alpha_m$ denote the asymmetric bracelets in $\A_k(n)$.  Let $\Set_k(n)$ denote the set $[\alpha_1] \cup [\alpha_2] \cup \cdots \cup [\alpha_m]$.

\begin{result}
\noindent
    {\bf Important property:}  If $\alpha \in \Set_k(n)$, then $\alpha^R \notin \Set_k(n)$.
\end{result}

\noindent
Let $L_k(n) = |\Set_k(n)|$.
Our main result defines a cycle-joining tree with nodes $\A_k(n)$, producing an $\OS_k(n)$ of length $L_k(n)$.  In Section~\ref{sec:bounds}, we provide a formula for $L_k(n)$ and demonstrate it to be an asymptotically optimal lower bound for $M_k(n)$ (the maximal length of an $\OS_k(n)$). The reminder of this section is devoted to properties of symmetric/asymmetric necklaces and bracelets required to prove our main results in Section~\ref{sec:parent}.

%


%
\begin{lemma}[\cite{G&S-Orientable:2024}] \label{lem:pal}
A necklace $\alpha$ is symmetric if and only if there exists palindromes $\beta_1$ and $\beta_2$ such that $\alpha = \beta_1 \beta_2$.
\end{lemma}
%
%

\begin{lemma} \label{lem:firstnon}
    Let $\alpha = \tt{a}_1\tt{a}_2\cdots \tt{a}_n \neq 0^n$ be a bracelet, where $i$ is the index of the first non-zero symbol.  If $\tt{a}_i > 1$, then $\firstNonMin(\alpha)$ is a bracelet; moreover, if $\alpha \in \A_k(n)$, then $\firstNonMin(\alpha) \in \A_k(n)$.  
\end{lemma}
\begin{proof}
    Suppose $\tt{a}_i > 1$.
        It is straightforward to apply the definition of a bracelet to verify that $\beta = \firstNonMin(\alpha) = 0^{i-1}(\tt{a}_i{-}1)\tt{a}_{i+1}\cdots \tt{a}_n$ is a bracelet, given $\alpha$ is a bracelet.  
        Let $\alpha \in \A_k(n)$.  Suppose $\beta$ is symmetric. From Lemma~\ref{lem:pal},  $\beta = \beta_1\beta_2$ for palindromes $\beta_1$ and $\beta_2$.  Clearly $|\beta_1| \geq i-1$.  Since $\alpha$ is a bracelet $\tt{a}_i \leq \tt{a}_n$, and thus
        $|\beta_1| \neq i-1$.  If $|\beta_1| = i$, then $i=1$ which implies that $\alpha$ is symmetric; otherwise, since $\beta_1$ has prefix $0^{i-1}$, it must have length at least $2i-1$.  
        If $|\beta_1| = 2i-1$, then $\alpha$ is symmetric; if  $|\beta_1| > 2i-1$ then $\alpha$ is not a bracelet. Each case contradicts that $\alpha \in \A_k(n)$. 
        Thus, $\beta$ is an asymmetric bracelet.
\end{proof}

\noindent
The following lemma follows by applying the definition of a bracelet.

\begin{lemma} \label{lem:lastnon}
If $\alpha \neq (k{-}1)^n$ is a bracelet,  then $\lastNonMax(\alpha)$ is a bracelet. 
\end{lemma}
%

\noindent
Note that $\lastSymbol(\alpha) = \lastNonMax(\alpha)$ when $\tt{a}_n < k{-}1$.

\begin{lemma} \label{lem:prop2}
Let $\alpha = \tt{a}_1\tt{a}_2\cdots \tt{a}_n \in \A_k(n)$, where $j$ is the index of the last non-($k{-}1$) and  $\ell$ is the index of the second last non-($k{-}1$).
    If (i) $\tt{a}_n < k{-}1$ or (ii) $\tt{a}_j < k{-}2$ or 
    (iii) $\tt{a}_1\cdots \tt{a}_{\ell} \neq (\tt{a}_1\cdots \tt{a}_{\ell})^R$, then $\lastNonMax(\alpha) \in \A_k(n)$;  otherwise, $\lastNonMax(\alpha)$ is a symmetric bracelet.
\end{lemma}
\begin{proof}
From Lemma~\ref{lem:lastnon}, $\beta = \lastNonMax(\alpha)$ is a bracelet.  If $\beta$ is symmetric, then from Lemma~\ref{lem:pal} it can be written as $\beta_1\beta_2$ for some palindromes $\beta_1$ and $\beta_2$.  Since $\tt{a}_1 < \tt{a}_n$ because $\alpha \in \A_k(n)$, we clearly have $\beta_1$ and $\beta_2$ are non-empty. 
\begin{itemize}
\item[(i)] Let $\tt{a}_n < k{-}1$. Suppose $\beta$ is symmetric. Let $\gamma$ be $\beta_2$ with its last symbol decremented; $\alpha = \beta_1\gamma$. Then $(\gamma \beta_1)^R \leq \alpha$, contradicting that $\alpha \in \A_k(n)$. 
\item[(ii)] Let $\tt{a}_j < k{-}2$.  Suppose $\beta$ is symmetric. 
If $|\beta_1| > j$, then it must start and end
with $(k{-}1)$ and include $\tt{a}_j$. This contradicts that $\beta$ is a bracelet.  If $|\beta_1| = j$, then since $\beta_1$ is a palindrome, the length $j$ prefix of $\alpha$ is greater than its reversal, contradicting that $\alpha$ is a bracelet.  Thus $\beta_2$ has suffix $(\tt{a}_j{+}1) \tt{a}_{j+1}\cdots \tt{a}_n$. If $|\beta_2| = 2(n-j)+1$, then the
suffix of $\alpha$ with the same length is also a palindrome, which implies $\alpha$ is symmetric, contradiction.  Otherwise, $|\beta_2| > 2(n-j)+1$.  Let $\gamma$ denote the suffix of $\alpha$ of length $|\beta_2|$.  It must be that $\gamma > \gamma^R$, which implies that $(\gamma\beta_1)^R < \alpha$, contradicting that $\alpha$ is a bracelet. Thus, $\beta \in \A_k(n)$.
\item[(iii)] Let $\tt{a}_1\cdots \tt{a}_{\ell} \neq (\tt{a}_1\cdots \tt{a}_{\ell})^R$. By the definition of a bracelet, $\tt{a}_1\cdots \tt{a}_{\ell} < (\tt{a}_1\cdots \tt{a}_{\ell})^R$.   If $\tt{a}_j < k{-}2$, then $\beta$ is in $\A_k(n)$ by the previous case.  Consider $\tt{a}_j = k{-}2$. Suppose $\beta  = \tt{a}_1\cdots \tt{a}_{\ell}(k{-}1)^{n-\ell}$ is symmetric.  We must have $|\beta_1| < \ell$. Since $\tt{a}_{\ell} < (k{-}1)$, $|\beta_2| \geq 2(n-\ell)+1$ .  But this implies that $\alpha$ is not a bracelet, contradiction. Thus, $\beta \in \A_k(n)$.
\end{itemize}
Finally, if $\tt{a}_n = k{-}1$, $\tt{a}_j = k{-}2$, and $\tt{a}_1\cdots \tt{a}_{\ell} = (\tt{a}_1\cdots \tt{a}_{\ell})^R$, then $\lastNonMax(\alpha) = \tt{a}_1\cdots \tt{a}_{\ell}(k{-}1)^{n-\ell}$, which is symmetric by Lemma~\ref{lem:pal}.
\end{proof}
%

%

\begin{lemma}  \label{lem:firstNonParent}
If $\alpha \in \A_k(n)$ such that 
$\lastNonMax(\alpha) \notin \A_k(n)$, $\lastSymbol(\alpha) \notin \A_k(n)$ and $\firstNonMin(\alpha) \in \A_k(n)$, 
then $\lastNonMax(\firstNonMin(\alpha)) \in \A_k(n)$ or $\lastSymbol(\firstNonMin(\alpha)) \in \A_k(n)$.
\end{lemma}
\begin{proof} 
Consider $\alpha =\tt{a}_1\tt{a}_2\cdots \tt{a}_n \in \A_k(n)$ where $i$ is the index of the first non-zero,  $j$ is the index of the last non-$(k{-}1)$, and $\ell$ is the index of the second-last non-$(k{-}1)$. 
Since $\lastNonMax(\alpha) \notin \A_k(n)$, from Lemma~\ref{lem:prop2}, $\tt{a}_j = k{-}2$, $\tt{a}_n = k{-}1$, and $\tt{a}_1\cdots \tt{a}_{\ell}= (\tt{a}_{1} \cdots \tt{a}_{\ell})^R$. 
Since $\alpha \in \A_k(n)$, 
$\alpha = \tt{a_1}\cdots \tt{a}_{\ell} (k{-}1)^{j-\ell-1}(k{-}2)(k{-}1)^{n-j}$ where $j-\ell-1 < n-j$.   
This implies that $\lastSymbol(\alpha) = 
0\tt{a}_1\cdots \tt{a}_{n-1}$ has prefix $0^i$; it is a necklace since $\alpha$ has no $0^{i}$ substring.
Since $\lastSymbol(\alpha) \notin \A_k(n)$,
$\tt{a}_i\cdots \tt{a}_{n-1} \geq (\tt{a}_i\cdots \tt{a}_{n-1})^R$, which means that $\tt{a}_i \geq \tt{a}_{n-1} \geq k{-}2$.
Let $\beta = \firstNonMin(\alpha)$.
If $i \geq n-1$, then $\alpha = 0^{n-2}(k{-}2)(k{-}1)$
and $\lastNonMax(\beta) = \alpha$.  Thus, assume $i < n-1$ and consider two cases.
\begin{itemize}
    \item Suppose $i > \ell$.  Then $i = \ell+1$ and $\alpha = 0^{\ell}(k{-}1)^{j-i}(k{-}2)(k{-}1)^{n-j}$.  Since $j-i < n-j$ and $\tt{a}_i\cdots \tt{a}_{n-1} \geq (\tt{a}_i\cdots \tt{a}_{n-1})^R$, $j-i = n-j-1 > 1$. Thus, $\lastSymbol(\beta) = 0^{\ell}(k{-}2)(k{-}1)^{n-i}$ which is in $\A_k(n)$.
    \item Suppose $i \leq \ell$.  If the first $\ell$ symbols of $\beta$ are not a palindrome, then by Lemma~\ref{lem:prop2}, $\lastNonMax(\beta) \in \A_k(n)$. Otherwise,  $\ell$ is odd, $i = (\ell+1)/2$, and $\alpha$ has prefix $0^{i-1}\tt{a}_i0^{i-1}$. 
    If $\tt{a}_i = k{-}2$, since $\tt{a}_i\cdots \tt{a}_{n-1} \geq (\tt{a}_i\cdots \tt{a}_{n-1})^R$, we have $\tt{a}_{n-1} = k{-}2$, $\alpha =  0^{i-1}(k{-}2)0^{i-1}(k{-}2)(k{-}1)$, and $\lastSymbol(\beta) = 0^{i}(k{-}3)0^{i-1}(k{-}2)$ which is in $\A_k(n)$. Otherwise, $\tt{a}_i = k{-}1$ and  
    $\tt{a}_{n-1} = k{-}1$ which means $i > 1$, $j < n-1$ and $\tt{a}_{n-2} \geq k{-}2$. But this contradicts that  $\tt{a}_i\cdots \tt{a}_{n-1} \geq (\tt{a}_i\cdots \tt{a}_{n-1})^R$.
\end{itemize}

\vspace{-0.1in}
\end{proof}

\section{A cycle-joining construction for orientable sequences}  \label{sec:parent}

In this section we provide a parent rule to create a cycle-joining tree with nodes $\A_k(n)$. We then apply the tree to  derive an $O(n)$-time successor rule for a corresponding universal cycle, which is an $\OS_k(n)$ of length $L_k(n)$.  

\subsection{A simple parent rule} \label{sec:parentrule}

The parent rule for the binary case defined in~\cite{G&S-Orientable:2024} uses $0^{n-4}1011$ as the root, and the parent of each non-root node $\alpha \in \A_2(n)$ is the first string in the list $\langle \firstNonMin(\alpha),\lastSymbol(\alpha),\lastNonMax(\alpha)\rangle$ that is also in $\A_2(n)$.  However, there are several issues when generalizing to a larger alphabet. In particular, the rule is not well-defined for $k=3$, and the corresponding cycle-joining tree does not have the Chain Property.  We will demonstrate each of these short-comings before deriving a new parent rule for alphabets of arbitrary size.

Assume $n,k \geq 3$. Let \blue{$\ktreeroot{n}$} $= 0^{n-2}(k{-}2)(k{-}1)$ denote the root of our upcoming cycle-joining tree.  The following example is for this specific root; however, similar examples exist for any arbitrary root.
%
\begin{exam} 
Consider any parent rule with root $r_{n,k}$ where the the parent of $\alpha \in \A_k(n)$ is the first string in a list starting with $\langle \firstNonMin(\alpha),\lastSymbol(\alpha), \ldots \rangle$  that is also in $\A_k(n)$. Let $\alpha = 01220\blue{2}$ and $\beta = 0\blue{1}0122$.  Then the parent of $\alpha$ is $\lastSymbol(\alpha) = 000122$ and 
the parent of $\beta$ is $\firstNonMin(\beta) = 000122$; $\alpha$ is joined  via conjugate pair ($201220, \blue{001220}$) and $\beta$ is joined via conjugate pair ($101220,\blue{001220}$). The two conjugate pairs share a string, and thus the corresponding cycle-joining tree does not have the Chain Property.    
\end{exam} \normalsize

The next example, and following lemma, demonstrate that the four functions $\firstSymbol$, $\firstNonMin$, $\lastSymbol$, and $\lastNonMax$ alone are not sufficient to define a parent rule with root $\ktreeroot{}$ when $k=3$.
%
\begin{exam} \label{exam:badone} 
 Consider $\alpha = 001022010012$ which is in $\A_3(12)$. No matter how we change the first, last, first non-zero, or last non-($k{-}1$) symbol in $\alpha$, the resulting string is not in $\A_3(12)$. In particular:
 \begin{itemize}
     
     \item $\firstSymbol(\alpha) = 0012\blue{2}0102201$ is not a bracelet,
    \item $\lastSymbol(\alpha) = \blue{0}00102201001$ is not a bracelet,
     \item $\firstNonMin(\alpha) = 00\blue{0}0022010012$ is not a bracelet, and
     \item $\lastNonMax(\alpha) = 0010220100\blue{2}2$ is symmetric.   
   \end{itemize}  
Such strings are uncommon. There are only 82 such strings in $\A_3(20)$ and they all have suffix 0012.    


\end{exam} \normalsize
\begin{lemma} \label{lem:counter}
Let $\alpha = 00102^{(n-10)} 010012$ for $n\geq 12$ and $k=3$.  Then $\alpha$ is in $\A_k(n)$ and each of
$\firstSymbol(\alpha)$, $\lastSymbol(\alpha)$, $\firstNonMin(\alpha)$, and $\lastNonMax(\alpha)$ is not in $\A_k(n)$.

\end{lemma}
\begin{proof}
It is straightforward to observe that $\alpha$ is a bracelet by definition and is asymmetric by Lemma~\ref{lem:pal}.  Applying the definitions, $\firstSymbol(\alpha) = 0012\blue{2}0102^{(n-10)}01$, $\lastSymbol(\alpha) = \blue{0}00102^{(n-10)}01001$,  and $\firstNonMin(\alpha) = 00\blue{0}002^{(n-10)}010012$ are all not bracelets, and 
$\lastNonMax(\alpha) = 00102^{(n-10)}0100\blue{2}2$ is symmetric.
\end{proof}

 Lemma~\ref{lem:counter}, demonstrates that for $k=3$, no parent rule exists for $\A_3(n)$ that applies only a combination of the four rules from Section~\ref{sec:cycle-join}. Thus, we define an additional function in order to define our parent rule for $\A_3(n)$. Let $\ell$ denotes the second last symbol in $\alpha$ that is not $(k{-}1)$, and define
\begin{itemize}
    \item $\secondLastNonMax(\alpha) = \tt{a}_1\cdots\tt{a}_{\ell-1}\blue{(\tt{a}_{\ell}{+}1)}\tt{a}_{\ell+1}\cdots \tt{a}_n$.
\end{itemize}
This function is well-defined, since each $\alpha  \in \A_k(n)$ must contain at least two symbols that are not $k{-}1$.
Recall Example~\ref{exam:badone}, where $\alpha =  001022010012$. Observe that $\secondLastNonMax(\alpha) = 001022010\blue{1}12$, which is in $\A_3(12)$.
%

%
%
\begin{result}
\noindent 
{\bf Parent rule for cycle-joining $\A_k(n)$ with root $\ktreeroot{n}$.}  ~~
If $\alpha = \tt{a}_1\tt{a}_{2}\cdots \tt{a}_n \in \A_k(n) \setminus \{\ktreeroot{} \}$,  
then define
\blue{$\parent(\alpha)$} to be the first string
that is an asymmetric bracelet in the list%
\[ \langle ~\lastNonMax(\alpha),~~ \lastSymbol(\alpha),~~ \firstNonMin(\alpha),~~ \secondLastNonMax(\alpha) ~\rangle. \] 

\vspace{-0.15in}
 \end{result}

%
\noindent
In the upcoming Lemma~\ref{lem:well-defined}, we demonstrate that $\parent(\alpha)$ is well-defined and interestingly, that $\secondLastNonMax(\alpha)$ is only necessary for $k=3$.
%
%
%
\noindent
The upcoming Theorem~\ref{thm:kcycle} demonstrates that the above parent rule induces a cycle-joining tree, which we denote by $\cycletree_k(n)$.  Moreover, we demonstrate $\cycletree_k(n)$ has the Chain Property in Theorem~\ref{thm:chain}.
In proving these results, 
we do not consider the case when $n=3$ and $k=3$, since $\A_3(3) = \{012\}$. Figure~\ref{fig:cycle63} illustrates $\cycletree_3(6)$. 

%



\begin{figure}[ht]
    \centering
    \resizebox{6.0in}{!}{\includegraphics{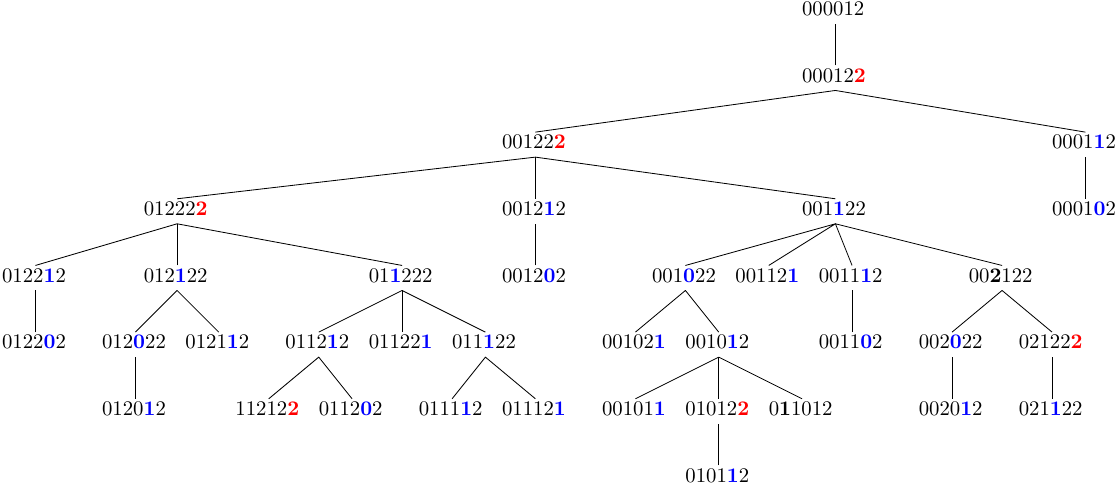}}
    \caption{The cycle-joining tree $\cycletree_{3}(6)$. 
    Each node differs from its parent (cyclically) at the highlighted symbol. The symbols highlighted in blue indicate that $\parent(\alpha) = \lastNonMax(\alpha)$; the symbols highlighted in red  indicate that $\parent(\alpha) = \lastSymbol(\alpha)$; the symbols highlighted in bold black indicate that $\parent(\alpha) = \firstNonMin(\alpha)$. 
    There are no nodes in this tree such that $\parent(\alpha) = \secondLastNonMax(\alpha)$; the first instance of such a case arises when $n=12$.}
    \label{fig:cycle63}
\end{figure}

\begin{lemma}\label{lem:well-defined}
   Let $\alpha = \tt{a}_1\tt{a}_2\cdots \tt{a}_n \in \A_k(n) \setminus \{r_{n,k}\}$ for some $n \geq 3$, $k \geq 4$ or $n\geq 4$, $k = 3$ such that $\lastNonMax(\alpha)$ and $\lastSymbol(\alpha)$ are not in $\A_k(n)$.  If $k \geq 4$, then   $\firstNonMin(\alpha) \in \A_k(n)$.
    Furthermore, if $k=3$ and $\firstNonMin(\alpha) \notin \A_3(n)$, then $\alpha=0\gamma012$ where $\gamma$ is a palindrome, and $\secondLastNonMax(\alpha) \in \A_3(n)$.
\end{lemma}    
\begin{proof}  
Consider $\alpha$ where $i$ is the index of the first non-zero,  $i'$ is the index of the second non-zero,  $j$ is the index of the last non-$(k{-}1)$, and $\ell$ is the index of the second-last non-$(k{-}1)$. 
Since $\lastNonMax(\alpha) \notin \A_k(n)$, from Lemma~\ref{lem:prop2}, $\tt{a}_j = k{-}2$, $\tt{a}_n = k{-}1$, and $\tt{a}_1\cdots \tt{a}_{\ell}= (\tt{a}_{1} \cdots \tt{a}_{\ell})^R$. 
Thus, since $\alpha \in \A_k(n)$, 
$\alpha = \tt{a_1}\cdots \tt{a}_{\ell} (k{-}1)^{j-\ell-1}(k{-}2)(k{-}1)^{n-j}$ where $j-\ell-1 < n-j$.   
 Thus, if $j=n-1$, then $\ell = n-2$.  This implies that $\lastSymbol(\alpha) = 
0\tt{a}_1\cdots \tt{a}_{n-1}$ with prefix $0^i$, which is a necklace since $\alpha$ has no $0^{i}$ substring.
Since $\lastSymbol(\alpha) \notin \A_k(n)$,
$\tt{a}_i\cdots \tt{a}_{n-1} \geq (\tt{a}_i\cdots \tt{a}_{n-1})^R$, which means that $\tt{a}_i \geq \tt{a}_{n-1} \geq k{-}2$.
If $k\geq 4$, then $\tt{a}_i > 1$, and thus by Lemma~\ref{lem:firstnon} $\firstNonMin(\alpha) \in \A_k(n)$.
If $k=3$, suppose $\firstNonMin(\alpha) \notin \A_k(n)$.
Then by Lemma~\ref{lem:firstnon}, $\tt{a}_i = 1$, and $\tt{a}_{i'}\cdots \tt{a}_n \geq (\tt{a}_{i'}\cdots \tt{a}_n)^R$.  Since $1 = \tt{a}_i \geq \tt{a}_{n-1}$, $j=n-1$ and $\tt{a}_{n-1} = (k{-}2) = 1$.
If $i=1$ then $\alpha$ does not contain any 0s, and thus $i'=2$. Thus $\tt{a}_2 = \tt{a}_n = 2$. 
The only bracelet that both starts and ends with 12 is of the form $(12)^s$ which is asymmetric, contradicting that $\alpha \in \A_k(n)$.
Thus $i > 1$ and $\tt{a}_1 = \tt{a}_{\ell} = 0$.
Recall $j=n-1$ which implied $\ell = n-2$.
Thus, $\alpha=0\gamma012$ where $\gamma$ is a palindrome. 
Let $\beta = \secondLastNonMax(\alpha)$; it has suffix 112.
Observe $\beta$ is a bracelet by applying the definition.
  Suppose $\beta$ is symmetric. By Lemma~\ref{lem:pal}, $\beta = \beta_1\beta_2$ where $\beta_1$ and $\beta_2$ are both palindromes. Clearly, $|\beta_2| > 3$.  Note that $\alpha$ does not have prefix 02, since $\alpha$ is a necklace containing 012 as a substring.
If $\beta_2 = 2112$, then $\alpha$ has prefix $02$, contradiction.  If $\beta_2 = 21112$, then $\alpha$ starts and ends with 012.  Since $\alpha$ is a necklace, this implies that $\alpha = (012)^t$ for some integer $t$, which contradicts the form of $\beta_2$. Thus $|\beta| > 5$. Let $\delta$ denote the suffix of $\alpha$ of length $|\beta_2|$. Then $\alpha$ is not a bracelet since $\delta > \delta^R$ and thus $\alpha^R < \alpha$. Contradiction.  Thus, $\beta$ is asymmetric and in $\A_3(n)$.
%
\end{proof}

Given a node $\alpha$ in $\cycletree_k(n)$, we say that $\gamma$ is an \defo{ancestor} of $\alpha$ if $\gamma = \alpha$ or $\gamma = \parent^t(\alpha)$ for some $t \geq 1$.

\begin{theorem}  \label{thm:kcycle}
   For $n,k \geq 3$, the parent rule $\parent(\alpha)$ for $\A_k(n)$ induces a cycle-joining tree $\cycletree_k(n)$ with nodes $\A_k(n)$ rooted at $\ktreeroot{n}$. 
\end{theorem}
\begin{proof}

Let $\alpha = \tt{a}_1\tt{a}_2\cdots \tt{a}_n \in \A(n) \setminus \{\ktreeroot{n}\}$, where $i$ is the index of the first non-zero, $j$ is the index of the last non-$(k{-}1)$, and $\ell$ is the index of the second last non-$(k{-}1)$ symbol in $\alpha$.
We demonstrate that $\ktreeroot{}$ is an ancestor of $\alpha$.
If $i=n-1$ then $\parent^{(k-1-\tt{a}_n) + (k-2-\tt{a}_{n-1})}(\alpha) = \ktreeroot{n} = 0^{n-2}(k{-}2)(k{-}1)$, where each application of $\parent$ uses $\lastNonMax$.
For $i \leq n-2$, we demonstrate that $\alpha$ has an ancestor $\beta$ such that $\beta < \alpha$, and thus must have an ancestor where $i=n-1$.  If $i=j$, then  $\alpha = 0^{i-1}\tt{a}_i(k{-}1)^{n-i}$ and $\parent^{k-1-\tt{a}_i}(\alpha) = 0^{i}(k{-}2)(k{-}2)^{n-i-1} < \alpha$.  Otherwise $i<j$.  If $k=3$, then from Lemma~\ref{lem:well-defined}, $\parent(\alpha) = \secondLastNonMax(\alpha)$ occurs only when $\alpha$ has suffix 012.  Thus, repeated applications of only $\lastNonMax$ and $\secondLastNonMax$ to $\alpha$  will never change any of the first $i$ symbols.  Since both operations only increment symbols to at most $(k{-}1)$, one of $\lastSymbol$ or $\firstNonMin$ must eventually be applied by repeated application of $\parent$ starting with $\alpha$.  If $\firstNonMin$ is applied then it will decrement $\tt{a}_i$ leading to an ancestor that is less than $\alpha$.  Similarly, $\lastSymbol$ will change the last symbol from $(k{-}1)$ to $0$ leading to a node with prefix $0^i$ which is also less than $\alpha$.
\end{proof}

By repeatedly joining cycles from $\cycletree_k(n)$ via conjugate pairs, we can construct an $\OS_k(n)$ (universal cycle) of length $L_k(n)$ using exponential time per symbol (to search for the conjugates) and exponential space.  Table~\ref{table:bounds} illustrates the lengths of the $\OS_k(n)$s constructed for small $n,k$. In Section~\ref{sec:bounds}, we present an exact formula for $L_k(n)$.

\begin{theorem}
There exists a universal cycle for $\Set_k(n)$, which is an $\OS_k(n)$, of length $L_k(n)$.
\end{theorem}

\begin{table}[h] \small
\begin{center} 
\begin{tabular}{c | r r r r r r}
\backslashbox{$n$}{$k$} & 3 & 4 & 5 & 6 & 7 & 8  \\ 
\hline 
 ~ 3 &             3 &           12 &           30 &           60 &          105 &          168  \\ 
 ~ 4 &            12 &           60 &          180 &          420 &          840 &         1512  \\ 
 ~ 5 &            60 &          360 &         1260 &         3360 &         7560 &        15120  \\ 
 ~ 6 &           225 &         1608 &         6750 &        21150 &        54831 &       124320  \\ 
 ~ 7 &           819 &         7308 &        36890 &       135450 &       403389 &      1034264  \\ 
 ~ 8 &          2676 &        30300 &       187980 &       821940 &      2844408 &      8315496  \\ 
 ~ 9 &          8778 &       126516 &       962580 &      5003970 &     20101326 &     66961608  \\ 
 ~10 &         27180 &       511680 &      4836300 &     30097620 &    140902440 &    536135040  \\ 
 ~11 &         84579 &      2074644 &     24328150 &    181141950 &    988016337 &   4293525544  \\ 
 ~12 &        257205 &      8327808 &    121790490 &   1087414170 &   6917824935 &  34352668560  \\ 
 ~13 &        782964 &     33447960 &    609843780 &   6528527460 &  48439152216 & 274864275504  \\ 
 ~14 &       2361177 &    133931952 &   3050119450 &  39175228260 & 339088485771 & 2198957209792  \\ 
 ~15 &       7125423 &    536379792 &  15255860130 & 235079896440 & 2373737520945 & 17592060218208  \\ 
 ~16 &      21419076 &   2146175580 &  76284577980 & 1410507942900 & 16616280850008 & 140736884449896  \\ 
 ~17 &      64402800 &   8587706400 & 381453125040 & 8463244062000 & 116314913988000 & 1125898765992000  \\ 
 ~18 &     193357350 &  34353845664 & 1907295914700 & 50779660926240 & 814205346309138 & 9007193819084160  \\ 
 ~19 &     580569795 & 137428992036 & 9536650390670 & 304679295576630 & 5699444909171768 & 72057583837380680  \\ 
 ~20 &    1742213832 & 549729612720 & 47683422899280 & 1828077103852860 & 39896121850134479 & 576460703985782112  \\ 
\end{tabular}
\end{center}
    \caption{ Lower bounds $L_k(n)$ on the maximal length of an $\OS_k(n)$ for $n \leq 20$  and $k \leq 8$.  }
    \label{table:bounds}
\end{table}

\begin{theorem}
    \label{thm:chain}
    $\cycletree_k(n)$ has the chain property. 

\end{theorem}
\begin{proof} 
By contradiction. 
Suppose $\cycletree_k(n)$ has a node $\gamma$ with two children $\alpha$ and $\beta$ joined via conjugate pairs $(\tt{x}\sigma, \tt{y}\sigma)$ and 
$(\tt{x}\sigma, \tt{y'}\sigma)$, respectively, for some string $\sigma$.  Note that $\alpha, \beta, \gamma \in \A_k(n)$.
Since $\alpha \neq \beta$, the functions they apply to obtain their parent $\gamma$ cannot both increment a symbol. Without loss of generality, assume $\parent(\beta) = \firstNonMin(\beta)$ and
$\parent(\alpha)$ applies one of $\lastNonMax$, $\lastSymbol$, or $\secondLastNonMax$.  
Let $\alpha = \tt{a}_1\tt{a}_2\cdots \tt{a}_n$ and consider the three possible cases for $\parent(\alpha)$.
\begin{itemize}
 \item    $\parent(\alpha) = \lastNonMax(\alpha)$. 
 Let $j$ denote the index of the last non-$(k{-}1)$ in $\alpha$. Since 
 $\gamma = \tt{a}_1\cdots \tt{a}_{j-1}(\tt{a}_j{+}1)\tt{a}_{j+1}\cdots \tt{a}_n$, $\beta = \tt{a}_1\cdots \tt{a}_{j-1}(\tt{a}_j{+}2)\tt{a}_{j+1}\cdots \tt{a}_n$ where $\tt{a}_j{+}2 \leq (k{-}1)$.  However, from Lemma~\ref{lem:firstnon}, either $\parent(\beta) = \lastNonMax(\beta)$, or $\beta$ is not in $\A_k(n)$.  Contradiction.

\item $\parent(\alpha) =\lastSymbol(\alpha)$. 
As noted in the proof of Lemma~\ref{lem:well-defined}, $\gamma = 0\tt{a}_1\cdots \tt{a}_{n-1}$.
Therefore $\beta = 1\tt{a}_1\cdots \tt{a}_{n-1}$.  Since $\beta \in \A_k(n)$, $1\tt{a}_1\cdots \tt{a}_{n-2} < (1\tt{a}_1\cdots \tt{a}_{n-2})^R$; otherwise, $\beta$ is not a bracelet or it is symmetric (see Lemma~\ref{lem:pal}).
Since $\parent(\beta) \neq \lastNonMax(\beta)$, from Lemma~\ref{lem:prop2}, $\tt{a}_{n-1} = k{-}1$. However, this means that 
$\lastSymbol(\beta) = 01\tt{a}_1\cdots \tt{a}_{n-2} \in \A_k(n)$, which contradicts that $\parent(\beta) = \firstNonMin(\beta)$.

\item $\parent(\alpha) =\secondLastNonMax(\alpha)$. From Lemma~\ref{lem:well-defined}, $\alpha$ has suffix 012 and thus $\gamma = \tt{a}_1\cdots \tt{a}_{n-3}112$. Therefore $\beta = \tt{a}_1\cdots \tt{a}_{n-3}212$ with its first non-0 at index $n-2$, which means $\beta$ is symmetric.  Contradiction.
 \end{itemize}
\end{proof}

\subsection{An $O(n)$-time successor rule} \label{sec:succ}

In this section, we apply the generic successor rule $g(\alpha)$ defined in Section~\ref{sec:successor} to the cycle-joining tree $\cycletree_k(n)$ that admits the Chain Property.  In particular, we determine whether or not $\alpha$ belongs to a conjugate pair, and if so, how to efficiently compute the function $\Call{Next}{\alpha}$.

Given $\alpha  = \tt{a}_1\tt{a}_2\cdots \tt{a}_n \in \Set_k(n)$, 
let $i$ be the largest index such that $\tt{a}_i > 0$, 
let $j$ be the smallest index greater than $1$ such that $\tt{a}_j < k{-}1$,  and let $\ell$ be the second smallest index greater than $1$ such that $\tt{a}_{\ell} < k{-}1$.  Recall that $\tilde \alpha$ denotes the necklace of $[\alpha]$.
If $\alpha$ belongs to some conjugate pair (possibly more than one) used to create $\cycletree_k(n)$, then we consider the possibilities for $\alpha$ depending on whether or not it belongs to a parent or child node joined by a given conjugate pair.  
If $\alpha$ belongs to the child, let $\beta = \tilde \alpha$.  If $\parent(\beta) = \lastNonMax(\beta)$, then $\beta = \tt{a}_j\tt{a}_{j+1}\cdots \tt{a}_n \blue{\tt{a}_1} (k{-}1)^{j-2}$. In other words, $\tt{a}_1$ corresponds to the last non-($k{-}1$) symbol in $\beta$.  If $\alpha$ belongs to the parent, let $\gamma$ denote the child node.  If $\parent(\gamma) = \lastNonMax(\gamma) = \tilde \alpha$, then it must be that $\gamma = \tt{a}_j\tt{a}_{j+1}\cdots \tt{a}_n \blue{(\tt{a}_1{-}1)} (k{-}1)^{j-2}$.  A similar analysis holds for the other three cases of the parent rule $\parent(\alpha)$ giving rise to the definitions of the following eight strings:
\begin{center}
    \begin{tabular}{lll}
    $\beta_1 =  \tt{a}_j\tt{a}_{j+1}\cdots \tt{a}_n \blue{\tt{a}_1} (k{-}1)^{j-2}$ &  ~~~ $\gamma_1 =  \tt{a}_j\tt{a}_{j+1}\cdots \tt{a}_n \blue{(\tt{a}_1{-}1)} (k{-}1)^{j-2}$ & ~~~~~~($\lastNonMax$) \\
    $\beta_2 = \tt{a}_2\tt{a}_3\cdots \tt{a}_n\red{\tt{a}_1}$ & ~~~ $\gamma_2 = \tt{a}_2\tt{a}_3\cdots \tt{a}_n\red{(\tt{a}_1{-}1)}$ & ~~~~~~($\lastSymbol$)\\ 

$\beta_3 = 0^{n-i}\bblack{\tt{a}_1}\tt{a}_2\cdots \tt{a}_i$ & ~~~ $\gamma_3 = 0^{n-i}\bblack{(\tt{a}_1{+}1)}\tt{a}_2\cdots \tt{a}_i$ & ~~~~~~($\firstNonMin$) \\ 

$\beta_4 =  \tt{a}_{\ell}\tt{a}_{{\ell}+1}\cdots \tt{a}_n \green{\tt{a}_1} \tt{a_2}\cdots \tt{a}_{{\ell}-1}$ & ~~~ $\gamma_4 =  \tt{a}_{\ell}\tt{a}_{{\ell}+1}\cdots \tt{a}_n \green{(\tt{a}_1{-}1)} \tt{a_2}\cdots \tt{a}_{{\ell}-1}$  &  ~~~~~~($\secondLastNonMax$).
\end{tabular}
\end{center}
Assume addition on the symbols is modulo $k$; i.e., $(k{-}1)+1=0$ and $0-1 = (k{-}1)$.  The above strings can be tested to determine whether or not $\alpha$ belongs to a conjugate pair. For instance if $\gamma_1 \in \A_k(n)$ and $\parent(\gamma_1) = \lastNonMax(\gamma_1)$, then $\alpha$ belongs to the conjugate pair $(\alpha, \blue{(\tt{a}_1{-}1)}\tt{a}_2\cdots \tt{a}_n)$.  If $\alpha$ is found to belong to some conjugate pair, then the second issue is to efficiently compute the function $\Call{Next}{\alpha}$.  If $\alpha$ belongs to some parent node of a conjugate pair, then $\Call{Next}{\alpha}$ is simply the incremented or decremented value of $\tt{a}_1$ defined by the parent rule for the corresponding $\gamma_i$.  If $\alpha$ does not belong to a parent in any conjugate pair, then it belongs to the last node in its corresponding chain which contains some $\beta_i$; to compute $\Call{Next}{\alpha}$ we must determine the first node in the chain. Na\"{i}vely, we can repeatedly check the ancestors $\alpha$ until we reach the top of the chain.  In the worst case this will take $O(kn)$ time.  With a deeper analysis of the four cases, we can remove the factor $k$. 
Let $\sigma = \tt{a}_j\tt{a}_{j+1}\cdots \tt{a}_n \blue{(k{-}2)} (k{-}1)^{j-2}$.
\begin{itemize}
    \item Suppose $\beta_1 \in \A_k(n)$ and $\parent(\beta_1) = \lastNonMax(\beta_1)$. Then Lemma~\ref{lem:prop2} implies that $\sigma \in \A_k(n)$ and $\Call{Next}{\alpha}$ is either $(k{-}2)$ or $(k{-}1)$, depending on whether or not $\parent(\sigma) = \lastNonMax(\sigma)$. 
    \item Suppose $\beta_2  \in \A_k(n)$ and $\parent(\beta_2) = \lastSymbol(\beta_2) = \delta$.
    Since $\lastNonMax(\beta_2) \notin \A_k(n)$, from Lemma~\ref{lem:prop2}, $a_1 = k{-}1$.  Furthermore, $\delta = 0\tt{a_2}\cdots\tt{a}_n$ since from Lemma~\ref{lem:prop2}, $\tt{a}_n$ is either $(k{-}2)$ or $(k{-}1)$.  It is easy to see that $\delta$ (if not the root) is not joined to $\parent(\delta)$ via a conjugate pair containing $0\tt{a_2}\cdots\tt{a}_n$. 
    Thus, $\Call{Next}{\alpha} = 0$.
    \item Suppose $\beta_3  \in \A_k(n)$ and $\parent(\beta_3) = \firstNonMin(\beta_3) = \delta$. 
Then Lemma~\ref{lem:firstNonParent} implies that
        $\parent(\delta) = \lastNonMax(\delta)$ or $\parent(\delta) = \lastSymbol(\delta)$.  Thus, $\delta$ is not joined to $\parent(\delta)$ via a conjugate pair containing the string
        $(\tt{a}_1{-}1)\tt{a}_2\cdots \tt{a}_n$ because it implies $\parent(\delta) = \tilde \alpha$.  Thus $\Call{Next}{\alpha} = \tt{a}_1 - 1$. 
        
    \item Suppose $\beta_4  \in \A_k(n)$ and $\parent(\beta_4) = \secondLastNonMax(\beta_4) = \delta$.  Then $k=3$ and by Lemma~\ref{lem:well-defined}, $\beta_4$ starts with 0 and has suffix $012$. This implies $\delta$ starts with 0 and has suffix 112.  By Lemma~\ref{lem:prop2}, $\lastNonMax(\delta) \in \A_3(n)$ and hence  $\parent(\delta) =  \lastNonMax(\delta)$.  Thus, $\Call{Next}{\alpha} = \tt{a}_1+1$.
\end{itemize}

\noindent
The above analysis gives rise to the following successor rule $h(\alpha) = g(\alpha)$ based on $\cycletree_k(n)$.

\begin{result}
\noindent
{\bf Successor-rule based on $\cycletree_k(n)$ to construct an $\OS_k(n)$ of length $L_k(n)$}  ~~Apply the conditions top down:

\medskip  \small

$h(\alpha) = \left\{ \begin{array}{ll}

         \tt{a}_1-1 &\ \  \mbox{if $\gamma_1 \in \A_k(n)$ and $\parent(\gamma_1) = \lastNonMax(\gamma_1)$;}\\
         \tt{a}_1-1 &\ \  \mbox{if $\gamma_2 \in \A_k(n)$ and $\parent(\gamma_2) = \lastSymbol(\gamma_2)$;}\\
         \tt{a}_1+1 & \ \ \mbox{if $\gamma_3 \in \A_k(n)$ and $\parent(\gamma_3) = \firstNonMin(\gamma_3)$;}\\
         \tt{a}_1-1 &\ \  \mbox{if $\gamma_4 \in \A_k(n)$ and $\parent(\gamma_4) = \secondLastNonMax(\gamma_4)$;}\\
         
& \\          
        k{-}1 &\ \  \mbox{if $\beta_1 \in \A_k(n)$ and $\parent(\beta_1) = \lastNonMax(\beta_1)$ and $\parent(\sigma) = \lastNonMax(\sigma)$;}\\
        k{-}2 &\ \  \mbox{if $\beta_1 \in \A_k(n)$ and $\parent(\beta_1) = \lastNonMax(\beta_1)$;}\\
        
         0 &\ \  \mbox{if $\beta_2 \in \A_k(n)$ and $\parent(\beta_2) = \lastSymbol(\beta_2)$;}\\

         \tt{a}_1-1 & \ \ \mbox{if $\beta_3 \in \A_k(n)$ and $\parent(\beta_3) = \firstNonMin(\beta_3)$;}\\

         \tt{a}_1+1 &\ \  \mbox{if $\beta_4 \in \A_k(n)$ and $\parent(\beta_4) = \secondLastNonMax(\beta_4)$;  }\\
 & \\
         {\tt{a}_1} \  &\ \  \mbox{otherwise.}\end{array} \right.$

\end{result}

\begin{theorem}
    For $n,k \geq 3$, $h(\alpha)$ is a successor rule for an $\OS_k(n)$ of length $L_k(n)$ that runs in $O(n)$ time and uses $O(n)$ space. 
\end{theorem}
\begin{proof}
Let $\alpha = \tt{a}_1\tt{a}_2 \cdots \tt{a}_n \in \Set_k(n)$.  Our previous analysis demonstrates that $h(\alpha) = g(\alpha)$ for the cycle joining tree $\cycletree_k(n)$.  Determining whether or not a string is a symmetric/asymmetric necklace or bracelet can be computed in $O(n)$ time and $O(n)$ space~\cite{booth,G&S-Orientable:2024}. Thus, all of the membership tester and functions required in the definition of $h(\alpha)$ can be computed in $O(n)$ time using $O(n)$ space.  
\end{proof}
In the next section, we demonstrate that length of the $\OS_k(n)$ generated by our successor rule $h(\alpha)$ is asymptotically optimal.

\section{Bounds on the maximal length of an $\OS_k(n)$}  \label{sec:bounds}
Recall that $L_k(n) = |\Set_k(n)|$ and $M_k(n)$ is the maximal length of an $\OS_k(n)$.  Our construction in Section~\ref{sec:parent} demonstrates that $L_k(n) \leq M_k(n)$ or $k \geq 3$.
Dai et al.~\cite{Dai} provide the following lower bound $L_2(n)$ for $M_2(n)$, where $\mu$ is the M\"{o}bius function:
\[  L_2(n) = \frac{1}{2}\left( 2^{n} -  \sum_{d  \mid  n} \mu(n/d) \frac{n}{d} H_2(d) \right), \ \ \ \text{ where } \ \ \ H_2(d) = \frac{1}{2} \sum\limits_{i  \mid  d} i \left( 2^{\lfloor \frac{i+1}{2} \rfloor} + 2^{\lfloor \frac{i}{2} \rfloor +1} \right).\]
Applying the same techniques, this formula can be  generalized to $L_k(n)$, 
\[ L_k(n) = \frac{1}{2}\left( k^{n} -  \sum_{d  \mid  n} \mu(n/d) \frac{n}{d} H_k(d) \right), \ \ \ \text{ where } \ \ \ H_k(d) = \frac{1}{2} \sum\limits_{i  \mid  d} i \left( k^{\lfloor \frac{i+1}{2} \rfloor} + k^{\lfloor \frac{i}{2} \rfloor +1} \right).\]
Exact values of $L_k(n)$ for some small $n,k$ are given 
in Table~\ref{table:bounds}.

\begin{theorem}
For $n \geq 3$ and $k \geq 3$, 
$ \displaystyle{ \lim_{n\to \infty} \frac{M_k(n) - L_k(n)}{L_k(n)} =0. }$ 
\end{theorem}

\begin{proof}    
Since $\mu(n/d) \leq 1$ and
$\displaystyle{ H_k(d)   \ = \    \frac{1}{2} \sum\limits_{i  \mid  d} i \left( k^{\lfloor \frac{i+1}{2} \rfloor} + k^{\lfloor \frac{i}{2} \rfloor +1} \right) \  \leq \    \sum\limits_{i=1}^d i  k^{\lfloor \frac{i}{2} \rfloor +1}  \  \leq \     d^2k^{\lfloor \frac{d}{2} \rfloor +1},}$
we have
\[ \sum_{d  \mid  n} \mu(n/d) \frac{n}{d} H_k(d)  \  \leq  \   \sum_{d=1}^n \frac{n}{d} d^2k^{\lfloor \frac{d}{2} \rfloor +1}  \ \leq \    n^3k^{\lfloor \frac{n}{2} \rfloor +1}.\]
Thus, 
$\displaystyle{ L_k(n) \ = \  \frac{1}{2}\left( k^{n} -  \sum_{d  \mid  n} \mu(n/d) \frac{n}{d} H_k(d) \right) \ \geq \  \frac{1}{2}\left( k^{n} -  n^3k^{\lfloor \frac{n}{2} \rfloor +1} \right). }$
Recalling from Section~\ref{sec:intro} that $M_k(n) \leq \frac{1}{2}(k^n - k^{\lfloor (n+1)/2\rfloor})$, we have 
\begin{align*}
\frac{M_k(n) - L_k(n)}{L_k(n)} & \  \leq \   
\frac{\frac{1}{2}(k^n - k^{\lfloor (n+1)/2\rfloor}) - \frac{1}{2}\left( k^{n} -  n^3k^{\lfloor \frac{n}{2} \rfloor +1} \right)}{\frac{1}{2}\left( k^{n} -  n^3k^{\lfloor \frac{n}{2} \rfloor +1} \right)}\\
& \  = \ 
\frac{  n^3k^{\lfloor \frac{n}{2} \rfloor +1}- k^{\lfloor (n+1)/2\rfloor}}{ k^{n} - n^3k^{\lfloor \frac{n}{2} \rfloor +1} }.
\end{align*}
The result follows.
\end{proof}

\bigskip

\begin{corollary}
    The $\OS_k(n)$ generated by the successor rule $h(\alpha$) has asymptotically optimal length.
\end{corollary}

\bibliographystyle{acm.bst}

\bibliography{abbrevs,refs}

\end{document}

%% file: paper_macros.tex
\usepackage[noend]{algpseudocode}
\usepackage{algorithm}
\usepackage{times}
\usepackage{amssymb}
\usepackage{amsmath}
\usepackage{latexsym}
\usepackage{graphics}
\usepackage{url}
\usepackage{verbatim}
\usepackage{color}
\usepackage{setspace}
\usepackage{multirow}
\usepackage{lineno}
\usepackage{booktabs}
\usepackage{graphicx}
\usepackage{caption}
\usepackage[shortlabels]{enumitem}
\usepackage{diagbox}

\algnewcommand\And{\textbf{and}}
\algnewcommand\Or{\textbf{or}}
\algnewcommand\Not{\textbf{not}}
\algnewcommand\In{\textbf{in}}
\algnewcommand\Each{\textbf{each}}

\newcommand{\squishlist}{
 \begin{list}{$\bullet$}
  { \setlength{\itemsep}{0pt}
     \setlength{\parsep}{3pt}
     \setlength{\topsep}{3pt}
     \setlength{\partopsep}{0pt}
     \setlength{\leftmargin}{2.5em}
     \setlength{\labelwidth}{1em}
     \setlength{\labelsep}{0.5em} } }

\newcommand{\squishlisttwo}{
 \begin{list}{$\triangleright$}
  { \setlength{\itemsep}{0pt}
     \setlength{\parsep}{0pt}
    \setlength{\topsep}{0pt}
    \setlength{\partopsep}{0pt}
    \setlength{\leftmargin}{2em}
    \setlength{\labelwidth}{1.5em}
    \setlength{\labelsep}{0.5em} } }

\newcommand{\squishend}{
  \end{list}  }

\usepackage{listings}

\definecolor{verbgray}{gray}{0.9}

\lstnewenvironment{code}{%
  \lstset{backgroundcolor=\color{verbgray},
 frame=single,
  language=C,
  framerule=0pt,
  basicstyle=\ttfamily,
  showstringspaces=false,
  commentstyle=\color{blue}\textit,
  keywordstyle=\color{black}\bf,
  numbers=none,
  keepspaces=true,
  columns=fullflexible}}{}

\definecolor{shadecolor}{rgb}{.91, .91, .91}
\definecolor{bordercolor}{rgb}{.8, .8, .6}

\definecolor{ultramarine}{rgb}{0, 0.125, 0.376}

 \definecolor{arsenic}{rgb}{0.23, 0.27, 0.29}
 \definecolor{beige}{rgb}{0.96, 0.96, 0.86}

\definecolor{amber}{rgb}{1.0, 0.75, 0.0}
\definecolor{orange}{rgb}{1.0, 0.49, 0.0}
\definecolor{dandelion}{rgb}{0.94, 0.88, 0.19}

  \definecolor{indiagreen}{rgb}{0.07, 0.53, 0.03}
  \definecolor{huntergreen}{rgb}{0.21, 0.37, 0.23}

\newcommand{\blue}[1] {\textcolor{blue}{#1}}
\newcommand{\red}[1] {\textcolor{red}{#1}}
\newcommand{\green}[1] {\textcolor{huntergreen}{#1}}

\newcommand{\bblack}[1] {{\bf \textcolor{black}{#1}}}
\newcommand{\bblue}[1] {{\bf \textcolor{blue}{#1}}}
\newcommand{\rred}[1] {{\bf \textcolor{red}{#1}}}

\newcommand{\defo}[1] {\emph{\textcolor{blue}{#1}}}

\definecolor{shadecolor}{rgb}{.9, .9, .9}

 \usepackage{framed}

    {\endMakeFramed}

    \newenvironment{frshaded*}{%
    \MakeFramed {\advance\hsize-\width \FrameRestore}}%
    {\endMakeFramed}
    
    \newcounter{examplecounter}
\newenvironment{exam}{
 \begin{frshaded*}
    \refstepcounter{examplecounter}%
    \noindent
  \textbf{Example \arabic{examplecounter}}%
  \quad
}{%
\end{frshaded*}
}

{\endMakeFramed}

\newenvironment{frshaded2*}{%
    \MakeFramed {\advance\hsize-\width \FrameRestore}}%
    {\endMakeFramed}

\newenvironment{result}{
 \begin{frshaded2*}
}{%
\end{frshaded2*}

}

{\endMakeFramed}

\newenvironment{frshaded3*}{%
    \MakeFramed {\advance\hsize-\width \FrameRestore}}%
    {\endMakeFramed}

\graphicspath{{graphics/}}